\documentclass{llncs}

\usepackage{makeidx} 
\usepackage{amsmath}
\usepackage{amssymb}

\usepackage{tikz}
\usepackage{pgfplots}
\usepackage{listings}
\usepackage{subfig}

\usepackage[section]{placeins}

\usepackage{changepage}

\pgfplotsset{compat=newest}

\newcommand{\mi}{\text{~~~~}}
\newcommand{\mii}{\text{~~~~~~~~}}
\newcommand{\miii}{\text{~~~~~~~~~~~~}}
\newcommand{\miiii}{\text{~~~~~~~~~~~~~~~~}}

\newcommand{\ty}{\texttt}

\newcommand{\sth}{\textsf{th}}
\newcommand{\sloc}{\textsf{loc}}
\newcommand{\srace}{\textsf{race}}
\newcommand{\sracefree}{\textsf{race-free}}
\newcommand{\slock}{\textsf{lock}}
\newcommand{\sunlock}{\textsf{unlock}}
\newcommand{\swr}{\textsf{wr}}
\newcommand{\srd}{\textsf{rd}}
\newcommand{\smem}{\textsf{mem}}
\newcommand{\shbrace}{\textsf{hb-race}}
\newcommand{\saarace}{\textsf{aa-race}}
\newcommand{\shb}{\textsf{hb}}
\newcommand{\sref}{\textsf{ref}}
\newcommand{\sterm}{\textsf{term}}
\newcommand{\seven}{\textsf{even}}

\newcommand{\smatch}{\textsf{match}}
\newcommand{\scheck}{\textsf{check}}
\newcommand{\sfinal}{\textsf{final}}
\newcommand{\ssrc}{\textsf{src}}
\newcommand{\stgt}{\textsf{tgt}}
\newcommand{\sconflict}{\textsf{conflict}}
\newcommand{\ssync}{\textsf{sync}}
\newcommand{\stype}{\textsf{type}}

\newcommand{\up}{\fontshape{n}}

\begin{document}

\pagestyle{headings} 
\mainmatter 

\title{Formalizing and Checking Thread Refinement for Data-Race-Free Execution Models\\\normalsize\vspace{2ex}(Extended Version)}

\titlerunning{Thread Refinement} 
\author{Daniel Poetzl \and
Daniel Kroening}

\institute{University of Oxford}

\maketitle 

\setcounter{footnote}{0}



\begin{abstract}
When optimizing a thread in a concurrent program (either done manually or by the
compiler), it must be guaranteed that the resulting thread is a refinement of
the original thread. Most theories of valid optimizations are formulated in
terms of valid syntactic transformations on the program code, or in terms of
valid transformations on thread execution traces. We present a new theory
formulated instead in terms of the state of threads at synchronization
operations, and show that it provides several advantages: it supports more
optimizations, and leads to more efficient and simpler procedures for refinement
checking. We develop the theory for the SC-for-DRF execution model (using locks
for synchronization), and show that its application in a compiler testing
setting leads to large performance improvements.
\end{abstract}


\section{Introduction}
\label{sec:introduction}

The refinement problem between threads appears in various contexts, such as the
modular verification of concurrent programs, the proving of correctness of
compiler optimization passes, or compiler testing.
Informally, a thread $T'$ is a refinement of a thread $T$ if for all possible
concurrent contexts $C = T_0\!\parallel\!\ldots\!\parallel\!T_{n-1}$ (with $\parallel$
denoting parallel composition), the set of final states reachable by
$T'\!\parallel\!C$ is a subset of the set of final states reachable by
$T\!\parallel\!C$. We consider the problem within the frame of code optimization
(either done manually or by an optimizing compiler): the optimized thread must
be a refinement of the original thread.


We consider refinement within the ``SC for DRF'' execution model~\cite{adve:1990}, i.\,e. programs behave sequentially
consistent (SC)~\cite{lamport:1979} if their SC executions are free of data races, and programs containing data
races have undefined semantics. A program containing data
races could thus end up in any final state. Synchronization is provided via lock$(l)$
and unlock$(l)$ operations. The model is similar to e.\,g. pthreads with the
various lock operations such as \ty{pthread\_mutex\_lock()} and 
\ty{pthread\_mutex\_unlock()}.

The definition of refinement given in the first paragraph is not
directly useful for automated or manual reasoning, as it would require the
enumeration of all possible concurrent contexts $C$.
We thus develop a new theory that is based on comparing the state of the
original thread and the transformed thread at synchronization operations.
We improve over existing work both
in terms of \emph{precision} and \emph{efficiency}. First, our theory allows to
show refinement in cases where others fail. For example, we also allow the
reordering of shared memory accesses out of critical sections (under certain
circumstances); a transformation that is unsupported by other theories.
Second, we show that applying our new specification method in a compiler testing
setting leads to large performance gains. We can check whether two thread
execution traces match significantly faster than a previous approach of Morisset
et al.~\cite{morisset:2013}.


The rest of the paper is organized as follows. Section~\ref{sec:vs} introduces
our state-based refinement formulation and compares it to previous event-based
approaches on a concrete example.
Section~\ref{sec:formalization} formalizes state-based refinement. Section~\ref{sec:optimizations} shows that our
formulation is more precise in that it supports more compiler optimizations than
current theories. Section~\ref{sec:nested_locks} adapts the formalization to
also handle nested locks.
Section~\ref{sec:evaluation} evaluates our theory in the context of a compiler testing
application that involves checking thread execution traces.
Section~\ref{sec:related} surveys related work. Section~\ref{sec:conclusions} concludes.



\section{State-Based and Event-Based Refinement}
\label{sec:vs}

Current theories of refinement for language level memory models (such
as the Java Memory Model or SC-for-DRF) are phrased in terms of transformations
on thread execution traces (see e.\,g. \cite{manson:2005,sevcik:2008,boehm:2007,morisset:2013,sevcik:2011}).
The trace transformations are then lifted to transformations on the program code.
Thread traces are sequences of memory events (reads or writes)
and synchronization events (lock or unlock).
The valid transformations are given as descriptions of which \emph{reorderings},
\emph{eliminations}, and \emph{introductions} of memory events on a trace are allowed.
Checking whether a trace $t'$ is a correctly transformed version of a trace $t$
then amounts to determining whether there is a sequence of valid
transformations that turns trace $t$ into trace $t'$. If each trace $t'$ of $T'$
is a transformed version of a trace $t$ of $T$, it follows that
$T'$ is a refinement of $T$.


We show that instead of describing refinement via a sequence of valid
transformations on traces, switching to a theory based on states provides several
benefits. In essence, in the state-based approach, we only require that traces
$t'$ and $t$ are in the same state at corresponding synchronization operations,
and that $t'$ does not allow for more data races than $t$. In the next
section, we illustrate the difference between the two approaches on an example.

\subsection{Example}

Consider Figure~\ref{fig:traces}, which shows an original thread $T$, a
(correctly) transformed version $T'$, and a concurrent context
$C$ in the form of another thread. The threads access shared variables $x, y, z$
and local variables $a, b$. The context $C$ outputs the value of variable $z$ in
the final state.
By inspecting $T'\!\parallel\!C$ and $T\!\parallel\!C$ (assuming initial state
$\{x\colon 0, y\colon 0, z\colon 0\}$), we see that both combinations produce
the same possible outputs ($0$ or $2$). In fact, $T'$ and $T$ exhibit the same
behavior in any concurrent context $C$ for which $T\!\parallel\!C$ is data-race-free.

Now let us look at two traces $t'$ of $T'$ and $t$ of $T$, and how an
event-based and our state-based theory would establish refinement. We assume for
now that $T$ and $T'$ are only composed with contexts that do not write any
shared memory locations accessed by them (as is the case for e.\,g. the context
shown in Figure~\ref{subfig:context}).
Figure~\ref{fig:trace_matching} shows the execution traces of $T$ (left trace)
and $T'$ (right trace) for initial state $\{x\colon 0, y\colon 0, z\colon 0\}$.


A theory based on trace transformations (see Figure~\ref{subfig:eventbased})
would establish the refinement between the two traces by noting that \ty{write x 2} can
be removed (``overwritten write elimination''), \ty{read x 2} and \ty{read y 0}
can be reordered (``non-conflicting read reordering''), and \ty{read y 0} can be
introduced (``irrelevant read introduction''). It can become significantly more
complicated if longer traces and more optimizations are considered.




\newsavebox{\context}
\begin{lrbox}{\context}
\begin{lstlisting}[frame=none]
void context() {
  int a;
  lock(l);
  a = x;
  z = a;
  unlock(l);
  join(thread_{orig|
               trans});
  printf("%d\n", z);
}
\end{lstlisting}%
\end{lrbox}

\newsavebox{\original}
\begin{lrbox}{\original}
\begin{lstlisting}[frame=none]
void thread_orig() {
  int a, b;
  lock(l);
  x = 1;
  x = 2;
  unlock(l);
  a = x;
  b = y;
  lock(l);
  if (b == 0)
    x = 0;
  unlock(l);
}
\end{lstlisting}
\end{lrbox}

\newsavebox{\transformed}
\begin{lrbox}{\transformed}
\begin{lstlisting}[frame=none]
void thread_trans() {
  int a, b;
  lock(l);
  x = 2;
  unlock(l);
  b = y;
  a = x;
  lock(l);
  if (b == 0)
    x = 0;
  b = y;
  unlock(l);
}
\end{lstlisting}
\end{lrbox}

\begin{figure}[t]
\hspace*{4mm}
\subfloat[Original thread\label{subfig:original}]{\usebox{\original}}%
\hspace*{5mm}
\subfloat[Transformed thread\label{subfig:transformed}]{\usebox{\transformed}}%
\hspace*{5mm}
\subfloat[Context\label{subfig:context}]{\usebox{\context}}
\caption{Original thread $T$, transformed thread $T'$, and concurrent context $C$}
\label{fig:traces}
\end{figure}



\newsavebox{\eventbased}
\begin{lrbox}{\eventbased}
\begin{tikzpicture}

\node[anchor=north west] at (0, 0) {%
\begin{lstlisting}[numbers=none, frame=none]
lock   m
write  x 1
write  x 2
unlock m
read   x 2
read   y 0
lock   m
write  x 0

unlock m
\end{lstlisting}
};

\node[anchor=north west] at (3, 0) {%
\begin{lstlisting}[numbers=none, frame=none]
lock   m
write  x 2

unlock m
read   y 0
read   x 2
lock   m
write  x 0
read   y 0
unlock m
\end{lstlisting}
};

\newcommand{\leftside}{1.9}
\newcommand{\rightside}{3}
\newcommand{\tracestep}{0.35}
\newcommand{\traceshift}{-0.1}
\draw[->, dashed, thick] (\leftside, \traceshift-6*\tracestep) -- (\rightside, \traceshift-5*\tracestep);
\draw[->, dashed, thick] (\leftside, \traceshift-5*\tracestep) -- (\rightside, \traceshift-6*\tracestep);

\node[anchor=north west] at (2.5, -3.2) {\tiny\textbf{(+)}};

\newcommand{\sllen}{1.85}
\newcommand{\sly}{-1.08}
\draw[thick] (0, \sly) -- (0+\sllen, \sly);

\end{tikzpicture}
\end{lrbox}


\newsavebox{\statebased}
\begin{lrbox}{\statebased}
\begin{tikzpicture}

\node[anchor=north west] at (0, 0) {%
\begin{lstlisting}[numbers=none, frame=none]
lock   m
write  x 1
write  x 2
unlock m
read   x 2
read   y 0
lock   m
write  x 0

unlock m
\end{lstlisting}
};

\node[anchor=north west] at (4.5, 0) {%
\begin{lstlisting}[numbers=none, frame=none]
lock   m
write  x 2

unlock m
read   x 2
read   y 0
lock   m
write  x 0
read   y 0
unlock m
\end{lstlisting}
};

\newcommand{\leftside}{1.9}
\newcommand{\rightside}{2.6}
\newcommand{\tracestep}{0.35}
\newcommand{\traceshift}{-0.1}
\draw[->, dashed, thick] (\leftside, \traceshift-3.9*\tracestep) -- (\rightside, \traceshift-3.9*\tracestep);
\draw[->, dashed, thick] (\leftside, \traceshift-10.6*\tracestep) -- (\rightside, \traceshift-10.6*\tracestep);

\renewcommand{\leftside}{4.4}
\renewcommand{\rightside}{3.8}
\renewcommand{\tracestep}{0.35}
\renewcommand{\traceshift}{-0.1}
\draw[->, dashed, thick] (\leftside, \traceshift-3.9*\tracestep) -- (\rightside, \traceshift-3.9*\tracestep);
\draw[->, dashed, thick] (\leftside, \traceshift-10.6*\tracestep) -- (\rightside, \traceshift-10.6*\tracestep);



\node[align=left] at (3.2, -1.5) {%
\scriptsize$\{x = 2,$\\\hspace{1.5mm}\scriptsize$y = 0,$\\\hspace{1.5mm}\scriptsize$z = 0\}$
};

\node[align=left] at (3.2, -3.8) {%
\scriptsize$\{x = 0,$\\\hspace{1.5mm}\scriptsize$y = 0,$\\\hspace{1.5mm}\scriptsize$z = 0\}$
};

\node[align=left] at (2, -5.6) {%
\scriptsize$R_0' \subseteq (A_0 \cup A_1)$\\
\scriptsize$W_0' \subseteq (W_0 \cup W_1)$\vspace{1mm}\\
\scriptsize$R_1' \subseteq A_1$\\
\scriptsize$W_1' \subseteq W_1$
};

\node[align=left] at (4.6, -5.1) {%
\scriptsize$R_2' \subseteq (A_2 \cup A_1)$\\
\scriptsize$W_2' \subseteq (W_2 \cup W_1)$
};

\end{tikzpicture}
\end{lrbox}

\begin{figure}
\subfloat[Event-based matching\label{subfig:eventbased}]{\usebox{\eventbased}}%
\hspace*{10mm}
\begin{minipage}{.6\linewidth}
\vspace*{-9mm}
\subfloat[State-based matching\label{subfig:statebased}]{\usebox{\statebased}}
\end{minipage}
\caption{Trace matching}
\label{fig:trace_matching}
\end{figure}



We specify trace refinement by requiring that $t'$, $t$ are in the
\emph{same state} at corresponding unlock operations, and that $t'$ does not
allow more data races than $t$ (see Figure~\ref{subfig:statebased}).
Indeed, both
traces are in state $\{x\colon 2, y\colon 0, z\colon 0\}$ at the first unlock$(l$), and in state
$\{x\colon 0, y\colon 0, z\colon 0\}$ at the second unlock($l$). The \emph{key reason} for why
trace refinement can be specified this way is that any context $C$ for which
$T\!\parallel\!C$ is data-race-free can for each shared variable only observe
the \emph{last write} to it before an unlock operation. If it could observe any
intermediate write, there necessarily would be a data race.

In addition to requiring that $t'$ and $t$ are in the same state, we also
require that $t'$ does not allow more data races than $t$.
This requirement is
specified by the set constraints in Figure~\ref{subfig:statebased}. The primed
sets correspond to $t'$, and the unprimed sets to $t$. The sets $R_i', R_i$
($W_i', W_i$) denote the sets of memory locations read (written) between
subsequent lock operations. For example, $R_1$ denotes the set of memory locations
read by $t$ between the first unlock($l$) and the second lock($l$). We also use
the abbreviations $A_i' = R_i' \cup W_i'$ and $A_i = R_i \cup W_i$. As an
example, the condition $W_0' \subseteq W_0 \cup W_1$ says that any memory
location written by $t'$ between the first lock($l$) and the subsequent
unlock($l$) must also be written by $t$ either between the first lock($l$) and
the subsequent unlock($l$), or between the first unlock($l$) and the subsequent lock($l$).
Since for $x \in W_0'$ we require only that $x \in W_0$ or $x \in W_1$, this
allows a write to move into the critical section in $t'$ compared to $t$. We
will more precisely capture the set constraints in Section~\ref{sec:formalization}.

\subsubsection{Contexts that Write}

We now assume that a thread can be put in an arbitrary context that can also
write to the shared state. Thus, when generating the traces of a thread we also
need to take into account that a read of a variable x could yield a value that
is both different from the initial value of x, and which the thread has not
itself written (i.e. it was written by the context).

In an event-based theory this is typically handled by assuming that reads can
return arbitrary values (see e.g. \cite{morisset:2013}). However, this
assumption is unnecessarily general. For example, if a thread reads the same
variable twice in a row with no intervening lock operation, and it
did not itself write to the variable, then both reads need to return the same
value. Otherwise, this would imply that another thread has written to the
variable and thus there would be a data race.

In fact, when generating the traces of a thread, it is sufficient to assume that
a thread observes the shared state only at its lock($l$) operations. The reason
for this is that lock($l$) operations synchronize with preceding unlock($l$)
operations of other threads. And those threads in turn make their writes
available at their unlock($l$) operations.

\subsubsection{State-Based Refinement}

To summarize, we state the intuitive formulation of our refinement theory. We
will formalize this notion in the next section.

\begin{quote}
We say that thread $T'$ is a refinement of thread $T$ if for each trace $t'$ of
$T'$ there is a trace $t$ of $T$ such that $t'$ and $t$ match.\\

We say two traces $t'$, $t$ match if their states at lock($l$) operations match
(i.\,e. they observe the same values), their states at unlock($l$) operations
match (i.\,e. they write the same values), and the sets of memory locations
accessed by $t'$ are subsets of the corresponding sets of memory locations
accessed by $t$ (i.\,e. $t'$ does not allow more data races than $t$).
\end{quote}

\section{Formalization}
\label{sec:formalization}

We now formalize the ideas from the previous section. We first
make a few simplifying assumptions. Most notably we assume for now that threads do not
contain nested locks. 
In Section~\ref{sec:nested_locks} we later adapt the formalization to also handle
nested locks.
%
We
assume that lock($l$) and unlock($l$) operations occur alternately on each thread
execution, and that lock($l$) and unlock($l$) operations occur infinitely often on any
infinite thread execution. This implies that a thread cannot get stuck e.\,g. in
an infinite loop without reaching a next lock operation. We also assume that the
first operation in a thread is a lock($l$), and the last \emph{lock} operation in a thread
is an unlock($l$). We assume that the concurrent execution is the only source of
nondeterminism, and that data races are the only source of undefined behavior.

A program $P = T_0\!\parallel\!\ldots\!\parallel\!T_{n-1}$ is a parallel composition
of threads $T_0, \ldots,\allowbreak T_{n-1}$. We denote by $h = (h_{T_0},
\ldots, h_{T_{n-1}})$ the vector of program counters of the threads. A program
counter (pc) points at the next operation to be executed. We use the predicate
$\textsf{lock}(T, h)$ (resp. $\textsf{unlock}(T, h)$) to denote that the next
operation to be executed by thread $T$ is a lock($l$) (resp. unlock($l$)). We
use $\textsf{term}(T, h)$ to denote that thread $T$ has terminated.

Let $M$ be a finite, fixed-size set of shared memory locations $x_1, \ldots,
x_{|M|}$. A state is a total function $s\colon M \rightarrow V$ from $M$ to the
set of values $V$. We denote the set of all states by $S$. We assume there is a
transition relation $\rightarrow$ between program configurations $(P, h, s)$. We
normally leave off $P$ when it is clear from context. The transition relation is
generated according to interleaving semantics, and each transition step
corresponds to an execution step of exactly one thread and accesses exactly one
shared memory location or performs a lock operation. We denote by $h_s = (h_{s, T_0}, \ldots,
h_{s, T_{n-1}})$ the initial pc vector with each thread at its entry point, and
by $h_f = (h_{f, T_0}, \ldots, h_{f, T_{n-1}})$ the final pc vector with each
thread having terminated.

We define a \emph{program execution fragment} $e$ as a (finite or infinite) sequence of
configurations such that successive configurations are related by $\rightarrow$.
A \emph{program execution} is an execution fragment that starts in a configuration with
pc vector $h_s$, and either has infinite length (i.e. does not terminate) or
ends in a configuration with pc vector $h_f$. A \emph{program execution prefix} is a
finite-length execution fragment that starts in a configuration with pc vector
$h_s$. Given an execution fragment such as $e = (h_0, s_0)(h_1, s_1)\ldots(h_n,
\allowbreak s_n)$, we use indices $0$ to $n-1$ to refer to the corresponding
execution steps. For example, index $0$ refers to the first execution step from
$(h_0, s_0)$ to $(h_1, s_1)$. We next define several predicates and functions on
execution fragments.

\vspace*{2ex}
\begin{tabular}{ll}
$\textsf{wr}(e, i)$: & step $i$ of $e$ is a shared write\\
$\textsf{rd}(e, i)$: & step $i$ of $e$ is a shared read\\
$\textsf{mem}(e, i)$: & $\swr(e, i) \vee \srd(e, i)$\\
$\textsf{conflict}(e, i, j)$: & $\textsf{loc}(e, i) = \textsf{loc}(e, j) \wedge (\textsf{wr}(e, i) \vee \textsf{wr}(e, j))$\\
$\textsf{lock}(e, i)$: & step $i$ of $e$ is a lock\\
$\textsf{unlock}(e, i)$: & step $i$ of $e$ is an unlock\\
$\textsf{loc}(e, i)$: & memory location/lock accessed by step $i$ of $e$\\
$\textsf{th}(e, i)$: & thread that performed step $i$ of $e$\\
$\textsf{src}(e, i)$: & source configuration of step $i$ of $e$\\
$\textsf{tgt}(e, i)$: & target configuration of step $i$ of $e$\\
$\textsf{initial}(e)$: & initial state of execution $e$\\
$\textsf{final}(e)$: & final state of execution $e$, or $\bot$ if $e$ is
infinite
\end{tabular}
\vspace*{2ex}

\noindent
We usually leave the execution $e$ off when it is clear from context. The
expression $\ssrc(e, i)$ (resp. $\stgt(e, i)$) refers to the configuration to the
left (resp. right) of $\rightarrow$ of the transition corresponding to step $i$
of $e$.

We next define the semantics of a program according to interleaving semantics as
the set of its initial/final state pairs. 

\begin{definition}[program semantics]
$\mathbb{M}(P) = \{(s, s')~|~$there exists an execution $e$ of $P$ such that
$|e| < \infty \wedge {\fontshape{n}\textsf{initial}}(e) = s \wedge
{\fontshape{n}\textsf{final}}(e) = s'\}$
\end{definition}





\noindent
Only finite executions are relevant for the program semantics as defined above.
Consequently, two programs $P'$, $P$ for which $\mathbb{M}(P') = \mathbb{M}(P)$
might have different behavior. For example, $P'$ might have a nonterminating
execution while $P$ might always terminate. The programs $P'$ and $P$ are thus
only \emph{partially equivalent}.

We next define the sequenced-before ($\textsf{sb}$), synchronizes-with
($\textsf{sw}$), and happens-before ($\textsf{hb}$) relation for a given
execution $e$ (with $|e| = n$). It holds that $(i, j) \in \textsf{sb}$ if $ 0
\le i < j < n$ and $\textsf{th}(i) = \textsf{th}(j)$. It holds that $(i, j) \in
\textsf{sw}$ if $0 \le i < j < n$, $\textsf{unlock}(i)$, $\textsf{lock}(j)$, and
$\textsf{loc}(i) = \textsf{loc}(j)$. The happens before relation $\textsf{hb}$
is then the transitive closure of $\textsf{sb} \cup \textsf{sw}$.

\begin{definition}[hb race]
We say an execution $e$ (with $|e| = n$) contains an \emph{hb data race},
written ${\fontshape{n}\shbrace}(e)$, if there are $0 \le i < j < n$ such that ${\fontshape{n}\sth}(i) \neq {\fontshape{n}\sth}(j)$,
${\fontshape{n}\sloc}(i) = {\fontshape{n}\sloc}(j)$, ${\fontshape{n}\swr}(i)$ or ${\fontshape{n}\swr}(j)$, and $(i, j) \notin {\fontshape{n}\shb}$.
\end{definition}

\begin{definition}[adjacent access race]
We say an execution $e$ (with $|e| = n$) contains an \emph{adjacent access data
race}, written ${\up\saarace}(e)$, if there are $0 \le i < j < n$ with $j - i
= 1$, ${\up\sth}(i) \neq {\up\sth}(j)$, ${\up\sloc}(i) = {\up\sloc}(j)$, and ${\up\swr}(i)$ or ${\up\swr}(j)$. 
\end{definition}

\noindent
The following lemma shows that these two data race definitions are equivalent
when they are lifted to the level of programs. For a proof see e.\,g. Boehm and
Adve~\cite{boehm:2008}.

\begin{lemma}
A program has an execution that contains an hb data race if and only if it has
an execution that contains an adjacent access data race.
\end{lemma}

\noindent
We write $\srace(P)$ to indicate that program $P$ has an execution that contains
a data race, and $\sracefree(P)$ to indicate that it does not have an execution
that has a data race. We are now in a position to define thread refinement.

\begin{definition}[refinement]\label{def:refinement}
We say that $T'$ is a refinement of $T$, written ${\up\sref}(T', T)$, if the following
holds:\\\vspace{-1ex}

\noindent
~~~~~~~~$\forall\,C \colon {\up\sracefree}(T\!\parallel\!C) \Rightarrow ({\up\sracefree}(T'\!
\parallel\!C) \wedge \mathbb{M}(T'\!\parallel\!C) \subseteq \mathbb{M}(T\!\parallel\!C))$
\end{definition}

\noindent
The definition says that for all contexts $C$ with which $T$ is data-race-free,
$T'$ is also data-race-free, and the set of initial/final state pairs of
$T'\!\parallel\!C$ is a subset of the set of initial/final state pairs of
$T\!\parallel\!C$.

The above definition is not directly suited for automated
refinement checking, as it would require implementing the $\forall$ quantifier
(and hence enumerating all possible contexts $C$). We thus develop in the
following our state-based refinement specification that implies $\sref(T', T)$,
and which is more amenable to automated and manual reasoning about refinement.

We next define the transition relation $\rightarrow^*$, which is more
coarse-grained than $\rightarrow$. It will form the basis of the refinement
specification.

\begin{definition}[$\rightarrow^*$]
$(P, h, s) \xrightarrow{l, (R_a, W_a), (R_b, W_b)}^* (P, h', s')$ if and only if
there exists an execution fragment $e = (h_0, s_0)(h_1, s_1),\allowbreak \ldots,
\allowbreak (h_k, s_k),\allowbreak \ldots, (h_n, s_n)$ such that ${\up\sth}(0) = {\up\sth}(1)
= \ldots = {\up\sth}(n-1) = T$ for some thread $T$ of $P$, ${\up\slock}(0)$,
${\up\smem(1)}, \ldots, {\up\smem(k-1)},{\up\sunlock}(k)$, ${\up\smem(k+1)}, \ldots, {\up\smem(n-1)}$,
either ${\up\slock}(T,\allowbreak h_n)$ or ${\up\sterm}(T, h_n)$, $\textsf{loc}(0) = l$, $h_0 = h$ and $h_n = h'$.
The set $R_a$ (resp. $W_a$) is the set of memory locations read (resp. written)
by steps $1$ to $k-1$. The set $R_b$ (resp. $W_b$) is the set of memory
locations read (resp. written) by steps $k+1$ to $n-1$.
\end{definition}


\noindent
We also use the abbreviations $A_a = R_a \cup W_a$ and $A_b = R_b \cup W_b$. The
relation $\rightarrow^*$ embodies uninterrupted execution of a thread $T$ of $P$ from a
lock($l$) to the next lock($l$) (or the thread terminates). Since we have excluded
nested locks, this means the thread executes exactly one unlock($l$) in between. For
example, in Figure~\ref{subfig:statebased} (left trace), the execution from the
first lock in line~1 to immediately before the second lock in line~7 corresponds
to a transition of $\rightarrow^*$. If we assume the thread starts in a state
with all variables being $0$, we have $s = \{x = 0, y = 0, z = 0\}$ and $s' =
\{x = 2, y = 0, z = 0\}$. The corresponding access sets are $R_a = \{\}, W_a =
\{x\}$, and $R_b = \{x, y\}, W_b = \{\}$.

We now define the semantics of a single thread $T$ as the set of its \emph{state
traces}. A state trace is a finite sequence of the form $(l_0, s_0, R_0, W_0)\allowbreak(R_1, W_1, s_1)
(l_2, s_2, R_2,\allowbreak W_2)\allowbreak(R_3, W_3, s_3)\allowbreak\ldots\allowbreak(l_{n-1}, s_{n-1}, R_{n-1}, W_{n-1})(R_n, W_n, s_n)$.
Two items $i$, $i+1$ (with $i$ being even) of a state trace belong together. The
item $i$ corresponds to execution starting in state $s_i$ at
a lock($l$) and executing up to the next unlock($l$), with the thread reading
the variables in $R_i$ and writing the variables in $W_i$. The subsequent item
$i+1$ corresponds to execution continuing at the unlock($l$) and executing until
the next lock($l$) reaching state $s_{i+1}$, with the thread reading the
variables in $R_{i+1}$ and writing the variables in $W_{i+1}$.

The formal definition of the state trace set $\mathbb{S}(T)$ is shown in
Figure~\ref{fig:state_traces}.
Intuitively, the state trace set of a thread $T$ embodies all interactions it
could potentially have with a context $C$ for which $\sracefree(T\!\parallel\!C)$.
A thread might observe writes by the context at a lock($l$) operation. This is
modeled in $\mathbb{S}(T)$ by the state changing between transitions. For
example, the target state $s_1$ of the first transition is different from the
source state $s_2$ of the second transition. The last line of the definition of
$\mathbb{S}(T)$ constrains how the state may change between transitions. It says
that those memory locations that the thread $T$ accesses in an execution portion
from an unlock($l$) to the next lock($l$) (i.\,e. those in $A_{i-1}$) do not
change at this lock($l$). The reason for this is that if those memory locations
would be written by the context then there would be a data race. But since
$\mathbb{S}(T)$ only models the potential interactions with race-free contexts,
the last line excludes those state traces.

Previously we stated that we are interested in the states of a thread at lock
and unlock operations, but $\mathbb{S}(T)$ embodies transitions from a lock$(l)$
to the next lock($l$). However, since we know the state at a lock($l$), and
we know the set of memory locations $W_i$ written between the previous unlock($l$)
and that lock($l$), we know the state of the memory locations $M-W_i$ at the
unlock($l$). This is sufficient for phrasing the refinement in the following.


\begin{figure}[t]
\normalsize
\begin{flalign*}
\mathbb{S}(T)\!=\!\{ &(l_0, s_0, R_0, W_0)(R_1, W_1, s_1)(l_2, s_2, R_2, W_2)(R_3, W_3, s_3)\ldots(R_n, W_n, s_n)\,|\\[1.5ex]
&\exists h_0, h_2, \ldots, h_{n+1}\colon\\
&(T, h_0, s_0) \xrightarrow{l_0, (R_0, W_0), (R_1, W_1)}^* (T, h_2, s_1)~\wedge\\
&(T, h_2, s_2) \xrightarrow{l_2, (R_2, W_2), (R_3, W_3)}^* (T, h_4, s_3)~\wedge\\
&\ldots\\
&(T, h_{n-1}, s_{n-1}) \xrightarrow{l_{n-1},(R_{n-1}, W_{n-1}), (R_n, W_n)}^* (T, h_{n+1}, s_n)~\wedge\\
&h_0 = h_s~\wedge\\
&\forall i \in \seven_n^+\colon \forall x \in A_{i-1}\colon s_{i-1}(x) = s_i(x) \}
\end{flalign*}
\caption{Definition of the state trace set of a thread}
\label{fig:state_traces}
\end{figure}

We are now in a position to define the $\smatch(t', t)$ predicate, which indicates
when a state trace $t' \in \mathbb{S}(T')$ matches a state trace $t \in
\mathbb{S}(T)$. The formal definition is shown in Figure~\ref{fig:match}. 
Primed symbols refer to components of $t'$, and unprimed symbols refer to
components of $t$. We denote by $\textsf{even}_n$ ($\textsf{odd}_n$) the set of
all even (odd) indices $i$ such that $0 \le i \le n$.
Intuitively, the constraints in lines $3$-$6$ specify that $t'$ must not allow
more data races than $t$. The constraints in lines $3$-$4$ correspond to an execution
portion from a lock($l$) to the next unlock($l$), and lines $5$-$6$ correspond to
an execution portion from the unlock($l$) to the next lock($l$). Since we have
$R_i' \subseteq A_{i-1} \cup A_i \cup A_{i+1}$ and $W_i' \subseteq W_{i-1}
\cup W_i \cup W_{i+1}$, the specification allows an access in $t$ to move into
a critical section in $t'$ (we further investigate this in Section~\ref{sec:optimizations}).
The constraint in line 7 specifies that $t'$ and $t$ receive the same new values
at lock($l$) operations (modeling writes by the context). The constraint at
line 9 specifies that the values written by $t'$ and $t$ before unlock($l$)
operations must be the same.
The last constraint specifies that $t'$ and $t$ perform the same sequence of lock
operations.

\begin{figure}[t]
\normalsize
\begin{flalign*}
&{\up\smatch}(t', t) \Leftrightarrow\\[1ex]
&\mi \text{\scriptsize 1}~~~{|u'| = |u|}\\
&\mi \text{\scriptsize 2}~~~{\textbf{let}~n = |u|~\textbf{in}}\\[1ex]
&\mi {\tt\#~race~constraints}\\
&\mi \text{\scriptsize 3}~~~\forall i \in \textsf{even}_n\colon R_i' \subseteq (A_{i-1} \cup A_i \cup A_{i+1})\\
&\mi \text{\scriptsize 4}~~~\forall i \in \textsf{even}_n\colon W_i' \subseteq (W_{i-1} \cup W_i \cup W_{i+1})\\
&\mi \text{\scriptsize 5}~~~\forall i \in \textsf{odd}_n\colon\hspace*{1mm} R_i' \subseteq A_i\\
&\mi \text{\scriptsize 6}~~~\forall i \in \textsf{odd}_n\colon\hspace*{1mm} W_i' \subseteq W_i\\[1ex]
&\mi {\tt\#~state~at~locks~constraints}\\
&\mi \text{\scriptsize 7}~~~\forall i \in \textsf{even}_n\colon \forall x \in M-A_{i-1}\colon s_i'(x) = s_i(x)\\
&\mi \text{\scriptsize 8}~~~\forall i \in \textsf{even}_n\colon \forall x \in A_{i-1}-A_{i-1}'\colon s_{i-1}'(x) = s_i'(x)\\[1ex]
&\mi {\tt\#~state~at~unlocks~constraints}\\
&\mi \text{\scriptsize 9}~~~\forall i \in \textsf{odd}_n\colon\hspace*{1mm} \forall x \in M-W_i\colon s_i'(x) = s_i(x)\\[1ex]
&\mi {\tt\#~same~locks~constraint}\\
&\mi \text{\scriptsize 10}~~~\forall i \in \textsf{even}_n\colon l_i' = l_i
\end{flalign*}
\vspace*{-2ex}
\caption{Definition of matching state traces}
\label{fig:match}
\end{figure}


We can now define our refinement specification $\scheck(T', T)$ which we later
show implies the refinement specification $\sref(T', T)$ of Definition~\ref{def:refinement}.
We denote by $t[0\colon\! i]$ the slice of a trace from index 0 to index $i$ (exclusive).

\begin{definition}[check]
\begin{flalign*}
&{\up\scheck}(T', T) \Leftrightarrow\\
&\mi \forall t' \in \mathbb{S}(T')\colon \exists t \in \mathbb{S}(T)\colon\\
&\mii {\up\smatch}(t', t) \vee\\
&\mii \exists i \in {\up\seven_n^+}\colon\\
&\miii {\up\smatch}(t'[0\colon\!i], t[0\colon\!i]) \wedge\\
&\miii \exists x \in (A_{i-1}-A'_{i-1})\colon s_{i-1}'(x) \neq s'_i(x)
\end{flalign*}
\end{definition}

\noindent
The definition says that either $t'$ and $t$ match, or there are prefixes
that match, and at the subsequent lock($l$) a memory location in $t'$ changes that is
accessed by $t$ but not by $t'$ ($x \in A_{i-1} - A_{i-1}')$. Thus, a context
that could implement the change of the memory location that $t'$ observes would
have a data race with $t$. Since when $t$ is involved in a data race we have
undefined behavior, any behavior of $t'$ is allowed. Hence, we consider the traces
$t'$ and $t$ matched.

We next state two lemmas that we use in the soundness proof of $\textsf{check}(T', T)$.

\begin{lemma}[coarse-grained interleaving]\label{lem:coarse}
Let $e$ (with $|e| = n$) be an execution prefix of $P$ with $\neg{\up\shbrace}(e)$ and ${\up\sfinal}(e) = s$.
Then there is an execution prefix $e'$ of $P$ with $\neg{\up\shbrace}(e')$ and ${\up\sfinal}(e') = s$, such
that execution portions from a {\upshape lock($l$)} to the next {\upshape lock($l$)} of a thread are not interleaved with other threads.
Formally:\\\vspace{-1.2ex}

\noindent
~~~~$\forall\,0 \le i < n\colon {\up\slock}(i) \Rightarrow \exists j > i\colon ({\up\slock}({\up\sth}(i), {\up\stgt(j)}) \vee {\up\sterm}({\up\sth}(i), {\up\stgt(j)})
\wedge\\~~~~~\forall i < k < j\colon {\up\sth}(k) = {\up\sth}(i))$
\end{lemma}

\renewcommand{\proofname}{Proof sketch}

\begin{proof}
Let $i$ be a step of $e$ with $\neg $lock($i$). Let $j$ be a step of $e$ with
$j < i$, $\sth(j) = \sth(i)$, such that $\forall j < k < i\colon \sth(k) \neq \sth(i)$.
It holds that $\forall j < k < i\colon (k, i) \notin \shb$.
Therefore, $\forall j < k < i\colon \neg \sconflict(k, i)$. Thus, step $i$ can be moved over the
steps $k$ and right after $j$ without changing the values read by any read
operation. We thus get a new execution prefix $e'$ with $\sfinal(e') = \sfinal(e)$.
Moreover, moving step $i$ cannot introduce a data race as it is moved
``upwards'' only and not past a lock($l$) operation.

The repeated application of picking a step $i$ with $\neg $lock($i$) and moving
it right after the previous step of the same thread finally yields an execution
prefix in which portions from a lock($l$) to the next lock($l$) are not
interleaved with other threads.

\qed
\end{proof}





\begin{lemma}[race refinement]\label{lem:race_soundness}
${\up\scheck}(T', T)\!\Rightarrow\!\forall\,C\!\colon ({\up\srace}(T'\!\parallel\!C)\!\Rightarrow\!
{\up\srace}(T\!\parallel\!C))$
\end{lemma}



\begin{proof}
Let $\srace(T'\!\parallel\!C)$. Then there is an execution that contains a data
race. A data race can either be between two threads in $C$, or between $T'$ and
a thread $C'$ in $C$. We assume the latter case. We further assume that
data races are between two writes on variable $x$. The other cases are analogous.

Since $\srace(T'\!\parallel\!C)$, there is an execution $e$ such that thread
$T'$ and thread $C'$ of $C$ are involved in an adjacent access data race.
Further, there is an (hb and adjacent access) race-free prefix $e'$ of $e$ such
that the next operation to be executed by each thread is a lock($l$), and the
next execution portions from a lock($l$) to the next lock($l$) of both $T'$ and
$C'$ are those involved in a data race.

Since the prefix $e'$ is data-race-free, by Lemma~\ref{lem:coarse} there is an
execution prefix $e''$ which ends in the same state as $e'$, and for which the
execution portions of a thread from a lock($l$) to the next lock($l$) are not
interleaved. Moreover, the execution of $T'$ and $C'$ can be continued from
$e''$ such that they are involved in an adjacent access data race. We denote
this continuation of $e''$ by $e'''$.

The sequence of execution portions of $T'$ in $e''$ corresponds to an element
$t' \in \mathbb{S}(T')$. The next execution portion of $T'$ from a lock($l$) to
the next lock($l$) after $e''$ is the one involved in the data race. Thus, $t'$
can be continued to $u' = t'(l_k', s_k', R_k', W_k')(R_{k+1}', W_{k+1}', s_{k+1}')$
such that $u' \in \mathbb{S}(T')$ and $x \in W_k' \cup W_{k+1}'$ (recall that we
assumed that data races are between two writes on variable $x$).

Then, by the definition of $\scheck(T', T)$, there is a $u \in \mathbb{S}(T)$
such that either (1) $\smatch(u', u)$ or (2)
$\exists i \in \textsf{even}_n\colon
\smatch(u'[0\colon\!i], u[0\colon\!i]) \wedge \exists x \in (A_{i-1} - A_{i-1}')
\colon s_{i-1}'(x) \neq s_i'(x)$.

\vspace*{2ex}
\begin{adjustwidth}{5mm}{}

(1) Let $u$ be of the form $t(R_{k-1}, W_{k-1}, s_{k-1})(l_k, s_k, R_k, W_k)
(R_{k+1}, W_{k+1}, s_{k+1})$. Since $t(R_{k-1}, W_{k-1}, s_{k-1})$ describes the
same state transitions as $t'$, the steps of $T'$ in $e''$ can be replaced by
steps of $T$. Then if this new execution prefix $q$ contains a data race we are
done (as we have $\srace(T\!\parallel\!C)$). We need to show that if the new
execution prefix $q$ does not contain a data race, then the next steps taken by
$T$ and $C'$ give rise to a data race.


By the definition of $\smatch(u', u)$, we have $W_k' \subseteq
W_{k-1} \cup W_k \cup W_{k+1}$ and $W_{k+1}' \subseteq W_{k+1}$.
In $e'''$, the access of $T'$ involved in the adjacent access data race might
either occur (a) between a lock($l$) and the subsequent unlock($l$) (i.\,e.
$x \in W_k'$), or (b)
between an unlock($l$) and the subsequent lock($l$) (i.\,e. $x \in W_{k+1}'$).

\vspace*{2ex}
\begin{adjustwidth}{5mm}{}

\noindent
(a) In this case the portion of $e'''$ containing the data race has the
following shape (portions denoted by an ellipsis ($\ldots$) contain only memory
accesses and no lock operations):\\

\noindent
\hspace{3ex}$\ldots,$ $T'\colon $lock($l$), $\ldots,$ $T'\colon W x$, $C'\colon W x$, $\ldots,$ $T'\colon$ unlock($l$), $\ldots$\\

\noindent
It further holds that $W_k' \subseteq W_{k-1} \cup W_k \cup W_{k+1}$. Thus, when
continuing to execute $T$ from $q$ a write to $x$ might occur either (i) before
the next lock$(l)$ ($x \in W_{k-1}$), (ii) between the next lock($l$) and unlock($l$) ($x
\in W_k$), or (iii) after the next unlock($l$) ($x \in W_{k+1}$).


\vspace*{2ex}
\begin{adjustwidth}{5mm}{}

\noindent
(i): In this case there is a continuation $q'$ of $q$ that contains an execution
fragment of the following form:\\

\noindent
\hspace{3ex}$\ldots,$ $T\colon W x$, $\ldots,$ $T\colon$ lock($l$), $\ldots,$ $C'\colon W x$, $\ldots$\\

\noindent
By the definition of the happens-before relation (\textsf{hb}), we see that
there is no \textsf{hb} edge between the steps ``$T\colon W x$'' and ``$C'\colon W x$''. Therefore,
there is a data race between the two steps.\\

\noindent
(ii): We have the following execution portion:\\

\noindent
\hspace{3ex}$\ldots,$ $T\colon$ lock($l$), $\ldots,$ $T\colon W x$, $\ldots,$ $C'\colon W x$, $\ldots$\\

\noindent
There is no $\shb$ edge between ``$T\colon W x$'' and ``$C'\colon W x$'', and thus there is a data race.\\

\noindent
(iii): We have the following execution portion:\\

\noindent
\hspace{3ex}$\ldots,$ $C'\colon W x$, $\ldots,$ $T\colon$ unlock($l$), $\ldots,$ $T\colon W x$, $\ldots$\\

\noindent
There is no $\shb$ edge between ``$C'\colon W x$'' and ``$T\colon W x$'', and thus there is a data race.\\

\end{adjustwidth}
\end{adjustwidth}

\vspace*{2ex}
\begin{adjustwidth}{5mm}{}

(b) In this case the portion of $e'''$ containing the data race has the following
shape:\\

\noindent
\hspace{3ex}$\ldots,$ $T'\colon$ unlock($l$), $\ldots,$ $T'\colon W x$, $C'\colon x$, $\ldots$\\

\noindent
It holds that $W_{k+1}' \subseteq W_{k+1}$. Thus, when continuing to execute $T$
from $q$ a write to $x$ occurs after the unlock($l$) just the same.
%
%
In this case there is a continuation $q'$ of $q$ that contains an execution
fragment of the following form:\\

\noindent
\hspace{3ex}$\ldots,$ $T\colon$ unlock($l$), $\ldots,$ $T\colon W x$, $\ldots,$ $C'\colon W x$, $\ldots$\\

\noindent
There is no $\shb$ edge between ``$T\colon W x$'' and ``$C'\colon W x$'', and thus there is a data race.

\end{adjustwidth}

\end{adjustwidth}

\vspace*{2ex}
\begin{adjustwidth}{5mm}{}
(2) Since $\smatch(u'[0\colon i], u[0\colon i])$, the first $i$ state transitions
described by $u$ are the same as those described by $u'$. Thus, we can
replace the first $i$ execution portions of $T'$ in $e''$ by execution portions of $T$. The last
execution portion of $T$ accesses a memory location $x$ that was not accessed by
the corresponding execution portion of $T'$ (since we have $\exists x \in
A_{i-1} - A_{i-1}'$). Moreover, by $s_{i-1}'(x) \neq s_i'(x)$ it follows that
this memory location is written by the context $C$. Thus, we have
$\srace(T \parallel C)$.
\end{adjustwidth}

\begin{adjustwidth}{0mm}{}
\qed
\end{adjustwidth}
\end{proof}

The following theorem establishes the soundness of our refinement specification
$\textsf{check}(T', T)$.

\begin{theorem}[soundness]\label{thm:soundness}
${\up\scheck}(T', T) \Rightarrow {\up\sref}(T', T)$
\end{theorem}

\begin{proof}
Let $C$ be an arbitrary context $C$ such that $\sracefree(T\!\parallel\!C)$. Let further
$(s, s')$ in $\mathbb{M}(T'\!\parallel\!C)$. Thus, there is an execution $e$ of
$T'\!\parallel\!C$ that starts in state $s$ and ends in state $s'$. By Lemma~\ref{lem:race_soundness},
$\sracefree(T'\!\parallel\!C)$. Thus, by Lemma~\ref{lem:coarse}, there is an execution $e'$ for which
portions from a lock($l$) to the next lock($l$) of a thread are not interleaved with
other threads. The sequence of those execution portions of $T'$ corresponds to
an element of $t' \in \mathbb{S}(T')$. Then, by the definition of
$\scheck(T', T)$, there is an element $t \in \mathbb{S}(T)$ such that either (a)
$\smatch(t', t)$, or (b) $\exists i \in \textsf{even}_n\colon
\smatch(t'[0\colon i], t[0\colon i]) \wedge \exists x \in (A_{i-1} - A_{i-1}')
\colon s_{i-1}'(x) \neq s_i'(x)$.

\vspace*{2ex}
\begin{adjustwidth}{5mm}{}
(a) Then $t$ embodies the same state transitions as $t'$. This is ensured by
constraints 7 and 9 of the definition of $\smatch()$. Constraint 7 specifies that
the starting states of a transition match, and constraint 9 specifies that the
resulting states of a transition match. Taking a closer look at constraints 7
and 9 reveals that the corresponding states of $t'$ and $t$ do not need to be
completely equal (only those memory locations in $M-A_{i-1}$ resp. $M-W_i$ need
to have the same value). The reason for this is that if a thread would observe
those memory locations it would give rise to a data race. Since we have both
$\sracefree(T' \parallel C)$ and $\sracefree(T \parallel C)$, it follows that the values of the
memory locations $A_{i-1}$ resp. $W_i$ can be arbitrary.
Therefore, $T$ can
make the same state transitions as $T'$. Thus, we can replace the steps of $T'$
in $e'$ by steps of $T$, and get a valid execution $e''$ of $T\!\parallel\!C$
ending in the same state. Therefore, $(s, s') \in \mathbb{M}(T\!\parallel\!C)$.
\end{adjustwidth}

\vspace*{2ex}
\begin{adjustwidth}{5mm}{}
(b) Since $\smatch(t'[0\colon i], t[0\colon i])$, the first $i$ state transitions
of $t$ are the same as those of $t'$. Thus, we can replace the first $i$
execution portions of $T'$ in $e'$ by execution portions of $T$. The last
execution portion of $T$ accesses a memory location $x$ that was not accessed by
the corresponding execution portion of $T'$ (since we have $\exists x \in
A_{i-1} - A_{i-1}'$). Moreover, by $s_{i-1}'(x) \neq s_i'(x)$ it follows that
this memory location is written by the context $C$. Thus, we have
$\srace(T \parallel C)$, which contradicts the premise $\sracefree(T \parallel C)$.
\end{adjustwidth}

\begin{adjustwidth}{0mm}{}
\qed
\end{adjustwidth}
\end{proof}

\section{Supported Optimizations}
\label{sec:optimizations}

We now investigate which optimizations are validated by our theory. By
inspecting the definition of $\textsf{match}()$ we see that it requires
that $t'$ and $t$ perform the same state transitions between lock operations,
and that the sets of memory locations accessed between lock operations of $t'$
must be subsets of the corresponding sets of memory locations accessed by $t$.
Together with the definition of $\textsf{check}()$, this implies that if an
optimization only performs transformations that do not change the state
transitions between lock operations, and does not introduce accesses to new
memory locations, then the optimized thread $T'$ will be a refinement of the
original thread $T$. 

Our theory also allows
the reordering of shared memory accesses into and out of critical sections
(under certain circumstances). The former are called roach motel reorderings and
have been studied for example in the context of the Java memory model (see
e.\,g.~\cite{sevcik:2008}). The latter have not been previously described in the
literature. In analogy to the former we term them \emph{inverse roach motel reorderings}.
We show on an example that both transformations are valid.


\subsubsection{Roach motel reorderings}

Consider Figure~\ref{fig:inverse}. Both $x$ and $y$ are shared variables.
Figure~\ref{subfig:originalrm} shows the
original thread $T$, and Figure~\ref{subfig:optimizedrmone} a correctly transformed
version $T'$. The statement \ty{y = 2} has been moved into the critical section. This
is save as it cannot introduce data races (but might remove data
races).


\newsavebox{\originalrm}
\begin{lrbox}{\originalrm}
\begin{minipage}{0.25\textwidth}
\begin{lstlisting}[frame=none]
lock(l);
x = 1;
y = 1;
unlock(l);
y = 2;
\end{lstlisting}
\end{minipage}
\end{lrbox}

\newsavebox{\optimizedrmone}
\begin{lrbox}{\optimizedrmone}
\begin{minipage}{0.25\textwidth}
\begin{lstlisting}[frame=none]
lock(l);
x = 1;
y = 1;
y = 2;
unlock(l);
\end{lstlisting}
\end{minipage}
\end{lrbox}

\newsavebox{\optimizedrmtwo}
\begin{lrbox}{\optimizedrmtwo}
\begin{minipage}{0.25\textwidth}
\begin{lstlisting}[frame=none]
lock(l);
x = 1;
unlock(l);
y = 1;
y = 2;
\end{lstlisting}
\end{minipage}
\end{lrbox}

\begin{figure}[t]
\hspace{10ex}
\subfloat[Original\label{subfig:originalrm}]{\usebox{\originalrm}}
\subfloat[Transformed 1\label{subfig:optimizedrmone}]{\usebox{\optimizedrmone}}
\subfloat[Transformed 2\label{subfig:optimizedrmtwo}]{\usebox{\optimizedrmtwo}}
\caption{Original, roach motel reordering, inverse roach motel reordering}
\label{fig:inverse}
\end{figure}

Let $t'$ be a state trace of $T'$ starting in some initial state $s_{init}$.
Then there is a state trace $t$ of $T$ starting also in $s_{init}$. The state
$s_{init}$ corresponds to the state at the first lock($l$) for both threads. At
the unlock($l$) they are in states $s' = \{x = 1, y = 2\}$ resp. $s = \{x = 1,
y = 1\}$. The access sets of the two state traces are $R_0' = R_1' = R_0 = R_1 =
\{\}$ (we ignore the read sets in the following as they are empty), and $W_0' =
W_0 = \{x, y\}, W_1' = \{\}, W_1 = \{y\}$.
At the unlock($l$), according to the definition of $\textsf{match}()$, the
constraint $\forall x \in M - W_1\colon s'(x) = s(x)$ needs to be satisfied.
This is the case as the variable $y$ for which $s'$ and $s$ differ is in $W_1$.
Moreover, for $\textsf{match}()$ to be satisfied, for the write sets the
following must hold: $W_0' \subseteq W_0 \cup W_1$ and $W_1' \subseteq W_1$.
This also holds. Hence, $\textsf{match}(t', t)$ holds. Consequently, we also
have $\textsf{check}(T', T)$ which implies $\textsf{ref}(T', T)$ according to
Theorem~\ref{thm:soundness}. $T'$ is thus a correctly transformed version of $T$.

\subsubsection{Inverse roach motel reorderings}

Consider now the example in Figure~\ref{subfig:optimizedrmtwo} which again shows
a correctly optimized version $T''$ of the thread $T$. In order to get defined
behavior of $T\!\parallel\!C$, the context $C$ must in particular avoid data races
with ${\tt y = 2}$. But this implies that the context cannot observe the write
${\tt y = 1}$, for if it could, there would be a data race with ${\tt y = 2}$.
Moreover, moving ${\tt y = 1}$ downwards out of the critical section cannot
introduce data races, as a write to $y$ already occurs in this section.
Consequently, ${\tt y = 1}$ can be moved downwards out of the critical section
(or in this particular case removed completely).

We can use a similar argument as in the previous section to show within our
theory that $T''$ is a
correctly optimized version of $T$. Let $t'$, $t$ be again two state traces starting in the same
initial state $s_{init}$. At the unlock($l$) they are in states $s' = \{x = 1,
y = y_{init}\}$ resp. $s = \{x = 1, y = 1\}$, with $y_{init}$ denoting the value
of $y$ in $s_{init}$. Again the constraints $\forall x \in M - W_1\colon s''(x)
= s(x)$, and $W_0'' \subseteq W_0 \cup W_1$ and $W_1'' \subseteq W_1$ are
satisfied, and we can conclude that $\textsf{match}(t', t)$, $\textsf{check}(T',
T)$, and finally $\textsf{ref}(T', T)$.


\section{Formalization with Nested Locks}
\label{sec:nested_locks}

We now adapt the formalization from Section~\ref{sec:formalization} to also allow nested
locks. To that end, we define a new coarse-grained transition relation
$\rightarrow_n$, the transition trace set $\mathbb{S}_n(T)$ of a thread, the
$\textsf{match}_n$ predicate, and finally the $\textsf{check}_n(T', T)$ predicate.

We first introduce some additional notation. We use $\ssync(T, h)$
($\Leftrightarrow \slock(T, h) \vee \sunlock(T, h)$) to indicate that the next
operation to be executed by thread $T$ is a lock operation or and unlock
operation. The function $\stype(e, i)$ returns the kind of step $i$ of execution
fragment $e$. This is one of $\textsf{lock}$, $\textsf{unlock}$, $\textsf{rd}$,
or $\textsf{wr}$. The predicate $\ssync(e, i)$ indicates that step $i$ of
execution fragment $e$ corresponds to a lock or unlock operation.

We now define the new coarse-grained transition relation $\rightarrow_n$. It
embodies execution from a lock operation to the next lock operation. Formally:




\begin{definition}[$\rightarrow_n$]
$(P, h, s) \xrightarrow{l, t, R, W}_n (P, h', s')$ if and only if there exists an
execution fragment $e = (h_0, s_0)(h_1, s_1)\ldots(h_n, s_n)$ such that
${\up\sth}(0) = {\up\sth}(1) = \ldots = {\up\sth}(n-1) = T$ for some thread $T$
of $P$, ${\up\ssync}(0)$, ${\up\smem}(1)$, $\ldots$, ${\up\smem}(n-1)$, ${\up\ssync}(T, h_n)$
or ${\up\sterm}(T, h_n)$, ${\up\sloc}(0) = l$, ${\up\stype}(0) = t$, $h_0 = h$ and
$h_n = h'$. The set $R$ (resp. $W$) is the set of memory locations read (resp.
written) by steps $1$ to $n-1$.
\end{definition}


The set $\mathbb{S}_n(T)$ denotes the \emph{transition trace set} of a thread
$T$. A transition trace has the form $(l_0, t_0, s_0, R_0, W_0, s_0^*)\allowbreak(l_1, t_1, s_1, R_1, W_1, s_1^*)\allowbreak\ldots\allowbreak(l_n, t_n,\allowbreak s_n,\allowbreak R_n,\allowbreak W_n,\allowbreak s_n^*)$. Each tuple corresponds
to a transition from a synchronization operation to immediately before the next
synchronization operation. The first component of a tuple denotes the lock
operated on, the second component denotes the type of the operation (either
\textsf{lock} or \textsf{unlock}), the third component denotes the starting
state, the fourth and fifth components denote the sets of memory locations read
or written, and the sixth component denotes the target state of the transition.

In Figure~\ref{fig:nextandprev} we define two predicates on transition traces.
Given a transition trace $t$ and an index $i$, they return the index of the
next transition that starts at a lock($l$), or the most recent transition
that started in an unlock($l$).


\begin{figure}
\normalsize
\begin{minipage}{0.5\textwidth}%
\begin{flalign*}%
&{\up\textsf{next-lock}}(t, i) = j \Leftrightarrow\\[1ex]
&\mi i < j\\
&\mi t_j = {\up\textsf{lock}}\\
&\mi \forall\, i \le k \le j\colon t_k = {\up\textsf{unlock}}\\
\end{flalign*}%
\end{minipage}%
\begin{minipage}{0.5\textwidth}%
\vspace*{-3.6ex}
\begin{flalign*}%
&{\up\textsf{prev-unlock}}(t, i) = j \Leftrightarrow\\[1ex]
&\mi j < i\\
&\mi t_j = {\up\textsf{unlock}}\\
&\mi \forall\, j \le k \le i\colon t_k = {\up\textsf{lock}}
\end{flalign*}%
\end{minipage}%
\caption{Next lock and previous unlock}
\label{fig:nextandprev}
\end{figure}


Figure~\ref{fig:state_tracesn} shows the transition trace set of a thread $T$.
Line~7 specifies that the state does not change at unlock operations, and lines
8-12 restrict how the state may change at lock operations.







\begin{figure}[t]
\normalsize
\begin{flalign*}
\mathbb{S}_n(T)\!=\!\{ & (l_0, t_0, s_0, R_0, W_0, s_0^*)(l_1, t_1, s_1, R_1, W_1, s_1^*)\ldots(l_n, t_n, s_n, R_n, W_n, s_n^*)\,|\\[1.5ex]
&\mi \text{\scriptsize 1}~~~~\exists h_0, \ldots, h_{n+1}\colon\\
&\mi \text{\scriptsize 2}~~~~\mi (T, h_0, s_0) \xrightarrow{l_0, t_0, R_0, W_0} (h_1, s_0^*)~\wedge\\
&\mi \text{\scriptsize 3}~~~~\mi (T, h_1, s_1) \xrightarrow{l_1, t_1, R_1, W_1} (h_2, s_1^*)~\wedge\\
&\mi \text{\scriptsize 4}~~~~\mi \ldots\\
&\mi \text{\scriptsize 5}~~~~\mi (T, h_n, s_n) \xrightarrow{l_n, t_n, R_n, W_n} (T, h_{n+1}, s_n^*)~\wedge\\
&\mi \text{\scriptsize 6}~~~~\mi h_0 = h_s~\wedge\\
&\mi \text{\scriptsize 7}~~~~\forall\, 0 < i \le n\colon (l_i = \textsf{unlock} \Rightarrow s_i = s_{i-1}^*)~\wedge\\
&\mi \text{\scriptsize 8}~~~~\forall\, 0 < i \le n\colon\\
&\mi \text{\scriptsize 9}~~~~\mi l_i = \textsf{lock} \Rightarrow\\
&\mi \text{\scriptsize 10}~~~~\mii \textbf{let}~j = \textsf{prev-unlock}(\text{this}, i)~\textbf{in}\\
&\mi \text{\scriptsize 11}~~~~\mii \forall x \in (M - (A_j \cup A_{j+1} \cup \ldots \cup A_{i-1}))\colon\\
&\mi \text{\scriptsize 12}~~~~\miii s_i(x) = s_{i-1}^*(x)\\
&\}
\end{flalign*}
\caption{Definition of the transition trace set of a thread}
\label{fig:state_tracesn}
\end{figure}


The $\textsf{match}_n(t', t)$ predicate between transition traces $t', t$ is
shown in Figure~\ref{fig:matchn}. The constraints correspond to those of the
previous $\textsf{match}(t', t)$ predicate for the case without nested locks.







\begin{figure}[t]
\normalsize
\begin{flalign*}
&{\up\smatch}_n(t', t) \Leftrightarrow\\[1ex]
&\mi \text{\scriptsize 1} \mi {|t'| = |t|}\\
&\mi \text{\scriptsize 2}\mi {\textbf{let}~n = |t|~\textbf{in}}\\[1ex]
&\mi {\tt\#~same~locks~constraint}\\
&\mi \text{\scriptsize 3}\mi \forall\,0 \le i \le n\colon l_i' = l_i \wedge t_i' = t_i\\[1ex]
&\mi {\tt\#~race~constraints}\\
&\mi \text{\scriptsize 4}\mi \forall\,0 \le i \le n\colon\\
&\mi \text{\scriptsize 5}\mii \textbf{let}~l = \textsf{prev-unlock}(i)~\textbf{in}\\
&\mi \text{\scriptsize 6}\mii \textbf{let}~j = \textsf{next-lock}(i)~\textbf{in}\\
&\mi \text{\scriptsize 7}\mii W_i' \subseteq (W_l \cup \ldots \cup W_i \cup \ldots \cup W_{j-1})\\
&\mi \text{\scriptsize 8}\mii R_i' \subseteq (A_l \cup \ldots \cup A_i \cup \ldots \cup A_{j-1}))\\[1ex]
&\mi {\tt\#~state~at~locks~constraints}\\
&\mi \text{\scriptsize 9}\mi \forall\,0 \le i \le n\colon\\
&\mi \text{\scriptsize 10}\mii l_i = \textsf{lock} \Rightarrow\\
&\mi \text{\scriptsize 11}\miii \textbf{let}~j = \textsf{prev-unlock}(i)~\textbf{in}\\
&\mi \text{\scriptsize 12}\miii \forall\,x \in (M- (A_j \cup A_{j+1} \cup \ldots \cup A_{i-1}))\colon\\
&\mi \text{\scriptsize 13}\miiii s_i(x) = s_i'(x)\\
&\mi \text{\scriptsize 14}\miii \forall\,x \in (A_j \cup A_{j+1} \cup \ldots \cup A_{i-1}) - (A_j' \cup A_{j+1}' \cup \ldots \cup A_{i-1}')\colon\\
&\mi \text{\scriptsize 15}\miiii s_i'(x) = s_{i-1}'^*(x)\\[1ex]
&\mi {\tt\#~state~at~unlocks}\\
&\mi \text{\scriptsize 16}\mi \forall\,0 \le i \le n\colon\\
&\mi \text{\scriptsize 17}\mii l_i = \textsf{unlock} \Rightarrow\\
&\mi \text{\scriptsize 18}\miii \textbf{let}~j = \textsf{next-lock}(i)~\textbf{in}\\
&\mi \text{\scriptsize 19}\miii \forall x \in (M - (W_i \cup W_{i+1} \cup \ldots \cup W_{j-1}))\colon\\
&\mi \text{\scriptsize 20}\miiii s_i(x) = s_i'(x)
\end{flalign*}
\vspace*{-2ex}
\caption{Definition of matching transition traces}
\label{fig:matchn}
\end{figure}


Finally, in Figure~\ref{fig:checkn} we define the $\textsf{check}_n(T', T)$
predicate. It implies $\sref(T', T)$ also in the case when $T'$ and $T$ contain
nested locks. 



\begin{figure}
\normalsize
\begin{flalign*}
&{\up\scheck}_n(T', T) \Leftrightarrow\\
&\mi \forall t' \in \mathbb{S}_n(T')\colon \exists t \in \mathbb{S}_n(T)\colon\\
&\mii {\up\smatch}(t', t) \vee\\
&\mii \exists\,0 \le i \le n\colon\\
&\miii t_i = \textsf{lock}~\wedge\\
&\miii {\up\smatch}(t'[0\colon\!i], t[0\colon\!i])\wedge\\
&\miii \textbf{let}~j = {\up\textsf{prev-unlock}}(i)~\textbf{in}\\
&\miii \exists x \in (A_j \cup A_{j+1} \cup \ldots \cup A_{i-1}) - (A_j' \cup A_{j+1}' \cup \ldots \cup A_{i-1}')\colon\\
&\miiii s_{i-1}'^* \neq s_i'(x)
\end{flalign*}
\caption{Check for threads with nested locks}
\label{fig:checkn}
\end{figure}

\section{Evaluation}
\label{sec:evaluation}

Previously we have argued that our specification efficiently captures thread
refinement in the SC-for-DRF execution model, as it abstracts over the way in
which a thread implements the state transitions between lock
operations. In this section we provide experimental evidence, showing that the
application of our state-based theory in a compiler testing setting leads to
large performance improvements compared to using an event-based theory.

Eide and Regehr~\cite{eide:2008} pioneered an approach to test that a
compiler correctly optimizes programs that involves repeatedly (1) generating a
random C program, (2) compiling it both with and without optimizations (e.\,g.
\ty{gcc -O0} and \ty{gcc -O3}), (3) collecting a trace from both the original
and the optimized program, and (4) checking whether the traces match. If two
traces do not match a compiler bug has been found. Morisset et al.~\cite{morisset:2013}
extended this approach to a fragment of C11 and implemented it in their
\textsf{cmmtest} tool.

The \textsf{cmmtest} tool consists of the following components: an adapted version of \textsf{csmith}~\cite{yang:2011} (we
call it ``\textsf{csmith-sync}'' in the following) to generate random C threads, a tool to collect execution traces of a thread
(``\textsf{pin-interceptor}''), and a tool to check whether two given traces match
(``\textsf{cmmtest-check}''). The \textsf{csmith-sync} tool generates random C threads
with synchronization operations such as
\ty{pthread\_mutex\_lock()}, \ty{pthread\_mutex\_unlock()}, or the C11 primitives
\ty{release()}
and \ty{acquire()}. We only consider programs containing lock operations. The
\textsf{pin-interceptor} tool is based on the Pin binary instrumentation framework~\cite{luk:2005}. It
executes a program and instruments the memory accesses and synchronization
operations in order to collect a trace of those operations. The \textsf{cmmtest-check}
tool takes two traces (produced by \textsf{pin-interceptor}) of an optimized and an
unoptimized thread, and checks whether the traces match.

\subsection{Implementation}
\label{subsec:implementation}

We use the existing \textsf{csmith-sync} and \textsf{pin-interceptor} tools, and implemented our
own trace checker \textsf{tracecheck}. It takes two traces (such as those
depicted in Figure~\ref{subfig:statebased}), and first determines the states of the traces at
lock operations, and the sets of memory locations accessed between
lock operations. That is, for a trace it constructs its
corresponding
state trace (i.\,e. an element of $\mathbb{S}(P)$). Then, it checks
whether the two state traces match by implementing the $\textsf{match}()$ predicate. This
way of checking traces is very efficient as it has runtime \emph{linear} in the
length of the traces.

This can be seen as follows. The size of a state is bounded by the number of
writes that have occured so far. Moreover, at each lock operation not the
complete states have to be checked for equality, but only the memory locations
that have been written to since the last check at the previous lock operation.
Thus, checking the states at lock operations (corresponding to the
``states at lock'' and ``states at unlock'' constraints of the $\smatch()$
predicate) is a linear operation.

The race constraints can also be checked in linear time. First, the size of the
sets is bounded by the number of memory locations accessed between the two
corresponding lock operations. Second, subset checking between two
sets $A$ and $B$ can be implemented in linear time. If $A$ and $B$ are represented as
hash sets, then $A \subseteq B$ can be checked by iterating over the elements of $A$, and
for each one performing a lookup in $B$ (which has constant time). If all
elements
are found, $A$ is a subset of $B$. In summary, we have a linear procedure for
checking whether two traces match.

\subsection{Experiments}
\label{subsec:experiments}

We evaluated \textsf{tracecheck} on in total $40,000$ randomly generated C
threads. We compiled each with \ty{gcc -O0} and \ty{gcc -O3} and collected a
trace from each. The length of the traces was in the range of 1 to 4,000
events. Our tool outperformed \textsf{cmmtest-check} on all traces. On average,
\textsf{tracecheck} was 210\,X faster.


Figure~\ref{fig:lentime} shows the average time it took to match two traces of a
certain length, for \textsf{cmmtest-check} (Figure~\ref{subfig:lentimecmmtest})
and \textsf{tracecheck} (Figure~\ref{subfig:lentimetc}). Along the x axis, we
classify the pairs of traces $t'$, $t$ into bins according to the length of the
unoptimized trace $t$. Each bin $i$ contains $100$ pairs $t'$, $t$ such that the
length of $t$ is in the range $[250\cdot i, 250 \cdot (i+1)]$. For example, bin $5$
contains the pairs with the length of the unoptimized trace being in the range
$[1250, 1500]$. The y axis shows the average time it took to match two traces
$t'$, $t$ in the respective bin. The dotted lines represent the 20th and 80th
percentile to indicate the spread of the times.


Figure~\ref{fig:lockstime} shows the effect of the number of lock operations in
the two traces on the time it takes to check if they match. We have evaluated
this on pairs of traces $t'$, $t$ with the unoptimized trace $t$ having length
in the range of $[1900, 2100]$. Along the x axis, we classify the pairs of
traces $t'$, $t$ into bins according to the number of
lock operations they contain. The y axis again indicates the average matching
time.
As can be seen in Figure~\ref{subfig:lockstimecmmtest}, \textsf{cmmtest-check}
is sensitive to the number of locks in a trace. That is, matching traces
generally takes longer the fewer locks they contain. The reason for this is that
\textsf{cmmtest-check} considers lock operations as ``barriers'' against
transformations: it does not try to reorder events across lock operations. Thus,
the more lock operations there are in a trace, the fewer potential
transformations it tries, and thus the lower the checking time.
Our tool \textsf{tracecheck} on the other hand is largely insensitive to the
number of locks in a trace.



\pgfplotsset{axisstyle1/.style={
  width=0.47\linewidth,
  height=0.27\textheight,
  ylabel={average checking time (in s)},
  xlabel={trace length (number of events)}
}}

\pgfplotsset{plotstyle1/.style={
  jump mark right,
  thick
}}

\pgfplotsset{varstyle1/.style={
  densely dotted
}}

\newsavebox{\lentimecmmtest}
\begin{lrbox}{\lentimecmmtest}
\begin{tikzpicture}[scale=1]
\small
\selectcolormodel{gray}
\begin{axis}[axisstyle1]
\addplot[plotstyle1, varstyle1] table[col sep=comma] {cmmtest_len-time-lower-cut.csv};
\addplot[plotstyle1]            table[col sep=comma] {cmmtest_len-time-cut.csv};
\addplot[plotstyle1, varstyle1] table[col sep=comma] {cmmtest_len-time-upper-cut.csv};
\end{axis}
\end{tikzpicture}
\end{lrbox}

\newsavebox{\lentimetc}
\begin{lrbox}{\lentimetc}
\begin{tikzpicture}[scale=1]
\small
\selectcolormodel{gray}
\begin{axis}[axisstyle1,
    y tick label style={
        /pgf/number format/.cd,
        fixed,
        fixed zerofill,
        precision=2,
        /tikz/.cd
    },
    scaled y ticks=false
]
\addplot[plotstyle1, varstyle1] table[col sep=comma] {tc_len-time-lower-cut.csv};
\addplot[plotstyle1]            table[col sep=comma] {tc_len-time-cut.csv};
\addplot[plotstyle1, varstyle1] table[col sep=comma] {tc_len-time-upper-cut.csv};
\end{axis}
\end{tikzpicture}
\end{lrbox}

\begin{figure}[t]
\hspace*{-2mm}
\subfloat[cmmtest\label{subfig:lentimecmmtest}]{\usebox{\lentimecmmtest}}%
\hspace*{3mm}
\subfloat[tracecheck\label{subfig:lentimetc}]{\usebox{\lentimetc}}
\caption{Average checking time over length of traces}
\label{fig:lentime}
\end{figure}


\pgfplotsset{axisstyle2/.style={
  width=0.47\linewidth,
  height=0.27\textheight,
  ylabel={average checking time (in s)},
  xlabel={Number of lock events}
}}

\pgfplotsset{plotstyle2/.style={
  jump mark right,
  thick
}}

\pgfplotsset{varstyle2/.style={
  densely dotted
}}

\newsavebox{\lockstimecmmtest}
\begin{lrbox}{\lockstimecmmtest}
\begin{tikzpicture}[scale=1]
\small
\selectcolormodel{gray}
\begin{axis}[axisstyle2]
\addplot[plotstyle2, varstyle2] table[col sep=comma] {cmmtest_locks-time-fixed-lower-cut.csv};
\addplot[plotstyle2]            table[col sep=comma] {cmmtest_locks-time-fixed-cut.csv};
\addplot[plotstyle2, varstyle2] table[col sep=comma] {cmmtest_locks-time-fixed-upper-cut.csv};
\end{axis}
\end{tikzpicture}
\end{lrbox}

\newsavebox{\lockstimetc}
\begin{lrbox}{\lockstimetc}
\begin{tikzpicture}[scale=1]
\small
\selectcolormodel{gray}
\begin{axis}[axisstyle2,
    y tick label style={
        /pgf/number format/.cd,
        fixed,
        fixed zerofill,
        precision=2,
        /tikz/.cd
    },
    scaled y ticks=false
]
\addplot[plotstyle2, varstyle2] table[col sep=comma] {tc_locks-time-fixed-lower-cut.csv};
\addplot[plotstyle2]            table[col sep=comma] {tc_locks-time-fixed-cut.csv};
\addplot[plotstyle2, varstyle2] table[col sep=comma] {tc_locks-time-fixed-upper-cut.csv};
\end{axis}
\end{tikzpicture}
\end{lrbox}

\begin{figure}[t]
\hspace*{-2mm}
\subfloat[cmmtest\label{subfig:lockstimecmmtest}]{\usebox{\lockstimecmmtest}}%
\hspace*{3mm}
\subfloat[tracecheck\label{subfig:lockstimetc}]{\usebox{\lockstimetc}}
\caption{Average checking time over number of locks in a trace}
\label{fig:lockstime}
\end{figure}


\section{Related work}
\label{sec:related}

Refinement approaches can be classified based on whether they handle
language-level memory models (such as SC-for-DRF or C11)~\cite{morisset:2013,%
boehm:2007,sevcik:2011,manson:2005,sevcik:2008}, hardware
memory models (such as TSO)~\cite{jagadeesan:2012,sevcik:2013}, or idealized
models (typically SC)~\cite{brookes:1993,lochbihler:2010}.

The approaches for language-level models typically describe refinement by
giving
valid transformations on thread execution traces. These trace transformations
are then lifted to the program code level. An example is the theory of valid
optimizations of Morisset et al.~\cite{morisset:2013}. They handle the fragment
of C11 with lock/unlock and release/acquire operations. The theory is
relatively
restrictive in that they do not allow the reordering of memory accesses across
synchronization operations (such as the roach model reorderings described in
Section~\ref{sec:optimizations}).

The approaches of Brookes~\cite{brookes:1993} (for SC) and
Jagadeesan~\cite{jagadeesan:2012} (for TSO) are closer to ours in that they
also
specify refinement in terms of state transitions rather than transformations on
traces. They provide a sound and complete denotational specification of
refinement. However, their completeness proofs rely on the addition of an
unrealistic await() statement which provides strong atomicity.

Liang et al.~\cite{liang:2012} presented a rely-guarantee-based approach to
reason about thread refinement. Starting from the assumption of arbitrary
concurrent contexts, they allow to add constraints that capture knowledge about
the context in which the threads run in. They later extended their approach to
also allow reasoning about whether the original and the refined thread exhibit
the same termination behavior~\cite{liang:2014}.

Lochbihler~\cite{lochbihler:2010} provides a verified non-optimizing compiler
for concurrent Java guaranteeing refinement between the threads in the source
program and the bytecode. It is however based on SC semantics rather than the
Java memory model. Sevcik et al.~\cite{sevcik:2013} developed the verified
CompCertTSO compiler for compilation from a C-like language with TSO semantics
to x86 assembly.

The compiler testing method based on checking traces of randomly generated
programs on which we evaluated our refinement specification in Section~\ref{sec:evaluation}
was pioneered by Eide and Regehr~\cite{eide:2008}. They used this approach to
check the correct compilation of volatile variables. It was extended to a fragment of C11
by Morisset et al.~\cite{morisset:2013}.

\section{Conclusions}
\label{sec:conclusions}

We have presented a new theory of thread refinement for the SC-for-DRF execution
model. The theory is based on matching the state of the transformed and the
original thread at lock operations, and ensuring that the former does
not introduce data races that were not possible with the latter. Our theory is
more precise than previous ones in that it allows to show refinement in cases
where others fail. Moreover, it supports efficient reasoning about refinement.


\newpage





\bibliographystyle{abbrv}
\bibliography{references}

\end{document}